\documentclass[journal,comsoc]{IEEEtran}
\usepackage{mathrsfs}
\usepackage{amsmath}
\usepackage{amssymb}
\usepackage{amsthm}
\usepackage[noadjust]{cite}
\usepackage{cite}
\usepackage{eurosym}
\usepackage{graphicx}
\usepackage{epstopdf}
\usepackage{tikz,lipsum}
\usepackage{subcaption}
\usepackage{float}
\usepackage[T1]{fontenc}
\usepackage{ragged2e}
\usepackage{multicol}
\usepackage{listings}
\usepackage{times}
\usepackage{array}
\usepackage[left=1.23cm,top=1.23cm,right=1.23cm,bottom=1.23cm]{geometry}
\usepackage{enumitem}
\usepackage{etoolbox}

\preto{\section}{\setcounter{subsection}{0}}
\newcounter{storeeqcounter}

\ifCLASSINFOpdf
\else
\fi

\bstctlcite{IEEEexample:BSTcontrol}
\ifCLASSOPTIONcompsoc
\else
\fi
\ifCLASSOPTIONcompsoc
\else
\fi
\ifCLASSINFOpdf
\else
\fi
\newtheorem{corollary}{Corollary}

\newtheorem{proposition}{Proposition}
\newtheorem{remark}{Remark}

\begin{document}
\title{\color{black}On the Secrecy of RIS-aided THz Wireless System subject to $\alpha-\mu$ fading with Pointing Errors}
\author[]{\color{black}Faissal~El~Bouanani, 
Elmehdi Illi, 
Marwa Qaraqe, 
 and Osamah Badarneh 
\thanks{\color{black}F. El Bouanani is with ENSIAS, Mohammed V University in Rabat, Morocco (e-mail:f.elbouanani@um5s.net.ma).}
\thanks{\color{black}E. Illi and M. Qaraqe are with the College of Science and Engineering, Hamad Bin Khalifa University, Doha, Qatar (e-mails: elmehdi.illi@ieee.org,mqaraqe@hbku.edu.qa)}
\thanks{Osamah Badarneh is with the Department of Electrical and Communication. Engineering, German Jordanian University, Madaba 11180, Jordan (e-mail: Osamah.Badarneh@gju.edu.jo)}
}
\maketitle
\begin{abstract}
\color{black}The study examines the secrecy outage probability (SOP) and intercept probability (IP) of a reflecting intelligent surface (RIS)-enabled THz wireless network experiencing $\alpha-\mu$ fading with pointing errors. Specifically, the base station (BS) sends information to a legitimate user $\ell$ via the RIS while an eavesdropper $e$ tries to overhear the conversation. Furthermore, receive nodes are equipped with a single antenna, and the RIS phase shifts were selected to boost the SNR at node $\ell$. Elementary functions are used to \textcolor{black}{accurately approximate} the statistical features of channel gain in BS-$\ell$ and BS-$e$ links, leading to SOP and IP \textcolor{black}{approximate and asymptotic} expressions. Monte Carlo simulation validates all analytical findings for different system parameters' values.
\end{abstract}
\begin{IEEEkeywords}
\color{black} THz communication, reconfigurable intelligent surfaces, secrecy outage probability, alternative Simpson's rule, Gauss–Laguerre quadrature.
\color{black}
\end{IEEEkeywords}
\IEEEpeerreviewmaketitle
\vspace{-0.25cm}\section{Introduction}
Terahertz (THz) communication technology has been gaining significant interest recently as an advocated means for achieving high data rates (in the order of Tbps) for the next generations' networks. Such technology relies on transmitting data in the $0.1-10$ THz frequency band using high-directivity beams \cite{akyildiz}. However, despite its merits, several impairments limit its performance, such as multipath fading, pointing errors (PE), molecular absorption, and atmospheric \textcolor{black}{attenuation. On top of that, hardware limitations, such as high power consumption, limited output power of THz sources, and costly fabrication of transceivers and RIS elements, remain major barriers.} From another front, reconfigurable intelligent surfaces (RIS) technology has been attracting notable interest as an enabler for significantly boosting networks' spectral efficiency and communication reliability. RIS is based on intelligent signal wave reflection through planar metasurfaces by optimally configuring the various unit reflective elements' phase shifts for signal beam steering to the desired user or area. Consequently, the amalgamation of RIS and THz technologies is expected to play a pivotal role in shaping the next generation's high data-rate networks \cite{basar}.

Physical layer security (PLS) has been widely advocated as a means for achieving keyless theoretically secure transmissions by utilizing wireless channel randomness, along with physical layer transmission schemes and channel coding. In this regard, RIS yields remarkable potential in achieving robust PLS schemes by properly beam steering the information signal towards the legitimate receiver or beamforming a noisy signal to the eavesdropper \cite{ILLI2024106}. \textcolor{black}{Although RIS enhances directionality and reduces signal leakage toward unintended users, secrecy remains a concern in THz systems due to their high sensitivity to blockage and misalignment, where even slight beam deviations or PE can unintentionally expose the signal to eavesdroppers}. Therefore, it is crucial to provide an analytical framework for quantifying the secrecy performance of RIS-aided THz wireless networks, considering THz channels' impairments. \textcolor{black}{It's important to note that compared to mmWave systems, RIS-assisted THz networks offer enhanced scalability, energy efficiency, and coverage, especially in dense or indoor environments, due to THz signals' ultra-wide bandwidth and high spatial resolution, coupled with RIS’s ability to dynamically mitigate severe path loss and support directional secure links with minimal hardware complexity. Nonetheless, RIS-assisted THz systems need precise CSI and rapid reconfigurability, making real-time deployment difficult. Research into CMOS-based THz circuits, low-cost metasurfaces, and intelligent beam tracking should overcome these restrictions and make THz communication more realistic.}
 
\textcolor{black}{Recently, RIS has been integrated with advanced communication paradigms like RIS-assisted Non-orthogonal multiple access (NOMA) with on-off control for secrecy enhancement \cite{10103182} and active simultaneously transmitting and reflecting surfaces to support NOMA frameworks \cite{10453463}, showing its versatility and evolving potential in secure wireless systems. On the other hand, various} works were reported in the performance assessment of the THz-based networks \textcolor{black}{from} either a reliability or security perspective. For instance, the authors in \cite{8610080}-\textcolor{black}{\cite{10159392}} inspected the reliability performance of THz communication systems, leveraging the $\alpha$-$\mu$ distributed fading model and Rayleigh distributed PE. In contrast, the authors of \cite{Chapala2021ExactAO} extended the analysis to non-identically distributed channels over both hops. Likewise, \cite{10159392} dealt with a RIS-aided THz network's reliability and energy efficiency analysis. From PLS's point of view, the work in \cite{9987659} provided an approximate ergodic secrecy capacity analysis of a THz-based downlink satellite-ground communication system. Furthermore, other work, such as \cite{ref1,ref2}, tackled optimizing the RIS phase shifts to maximize the THz networks' secrecy capacity using convex optimization tools.

The works above were mostly restricted to the reliability analysis of RIS-assisted THz networks. Furthermore, other work, such as \cite{9987659}, employed an approximate expression for the ergodic secrecy capacity (ESC), subject to the Gamma-Gamma turbulence fading model with PE, which exhibits some gaps with respect to the exact ESC value. Motivated by the above, this paper proposes an accurate analytical secrecy evaluation framework for RIS-assisted THz wireless communication \textcolor{black}{systems} (WCS) \textcolor{black}{with perfect channel state information (CSI)}, subject to the various THz channel impairments. Particularly, tightly approximate \textcolor{black}{and asymptotic} expressions for the secrecy outage probability (SOP) and intercept probability (IP) for RIS-aided THz WCS, experiencing $\alpha-\mu$ turbulence-induced fading with PE, are derived.\vspace{-0.25cm} 	 
  \color{black}
\section{\textcolor{black}{System model}}
  \color{black}
We consider \textcolor{black}{an} RIS-aided THz WCS subject
to $\alpha-\mu$ turbulence fading with PE misalignment
at \textcolor{black}{the} receivers. A source $\mathbf{S}$ sends signals in the THz band to
a legitimate user $\ell$ while an eavesdropper $e$ attempts to overhear the exchange. Furthermore, end users are equipped with single antennas.
Yet, the received signal may be represented as follows\vspace{-0.15cm}
\begin{equation}
y_{x}=\sqrt{\mathcal{P}}h_{l,x}h_{x}s+n_{x};x\in\{\ell,e\},\label{receiv}
\end{equation}
where $\mathcal{P}$ is the transmit power, $s$ as the transmitted signal with $\mathbb{E}\left[\left\vert s\right\vert ^{2}\right]=1$$,\mathbb{E}[\cdot]$
as the expectation operator, $h_{l,x}$ as the path
gain from $\mathbf{S}$ to $x$,
$h_{x}$ as the channel gain with PE, and $n_{x}\sim\mathcal{CN}\left(0,\sigma_{x}^{2}\right)$ is the
additive white Gaussian noise (AWGN) at the node $x$. Thus, the SNR at node $x$ can be expressed as
\begin{equation}
\gamma_{x}=\frac{\mathcal{P}h_{l,x}^{2}\left\vert h_{x}\right\vert ^{2}}{\sigma_{x}^{2}},\label{eq:SNR}
\end{equation}\vspace{-.25cm}
with CDF\vspace{-0.15cm}
\begin{equation}
    F_{\gamma_{x}}(x)=F_{\left|h_{x}\right|^{2}}\left(\frac{z\sigma_{x}^{2}}{\mathcal{P}h_{l,x}^{2}}\right),z\ge 0.\label{CDFgamma}
\end{equation}\vspace{-0.65cm}
\subsection{Path gain}
The path gain $h_{l,x}$ is a constant that depends on the node $x$'s location and the system's parameters in terms of FRIIS and molecular absorption terms as \cite{8610080} $h_{l,x}=h_{l,x}^{(F)}h_{l,x}^{(a)},$ where $h_{l}^{(F)}$, in the case of RIS-assisted WCS, can be modeled in the near-field broadcasting as  $h_{l,x}^{(F)}=\frac{\lambda\sqrt{G_{t}G_{r,x}}}{4\pi d_{x}}$ \cite{9206044}, with $\lambda$ as the wavelength, $G_{t}$ and $G_{r,x}$ as the transmit
and receive antenna gain, $d_{x}$ as the distance distance $\mathbf{S}-x$,
i.e., $d_{x}=d_{1}+d_{r,x}$ where $d_{1},d_{r,x}$ are the distance
$\mathbf{S}$-RIS and RIS-$x$,
$h_{l,x}^{(a)}=\exp\left(-{\kappa_{a}\left(f\right)}d_{x}/{2}\right),$ and $\kappa_{a}\left(\cdot\right)$ denote the absorption coefficient depending on the signal frequency \cite{8610080}.
\color{black}\vspace{-.25cm}
\subsection{Channel gain with PE}
The channel gain with PE, defined in \eqref{receiv} is expressed as
\begin{equation}
h_{x}=h_{f,x}h_{p,x},\label{hx}
\end{equation}\vspace{-.25cm}
where\vspace{-.25cm}
\begin{align}
h_{f,x} & =\sum\limits _{n=1}^{N}h_{n}g_{n,x}\exp\left(j\theta_{n}\right)\label{eq:hxf},
\end{align}
is the cascaded fading between the source $\mathbf{S}$
and $x$ and $h_{p,x}$ is the fading attenuation due to PE, with $h_{n}$ and $g_{n,x}$ denote the channel coefficients
$\mathbf{S}-\mathbf{r}_{n}$ and $\mathbf{r}_{n}-x$, $\mathbf{r}_{n}$
as the $n$th reflecting element of RIS, and $\theta_{n}$ is the
phase shift at $\mathbf{r}_{n}$. Considering the optimal RIS phase shifts, maximizing the received SNR at the legitimate user as $\theta_{n}=-\angle h_{n}-\angle g_{n,\ell}$ yields\vspace{-.2cm}
\begin{align}
h_{f,\ell} & =\left|h_{f,\ell}\right|\triangleq \sum\limits _{n=1}^{N}\left|h_{n}\right|\left|g_{n,\ell}\right|,\label{eq:hfl}
\end{align}\vspace{-.25cm}
and
\begin{align}
h_{f,e} \triangleq h_{f,e}^{(c)}+jh_{f,e}^{(s)},\label{eq:hxf-1}
\end{align}
\begin{equation}
h_{f,e}^{(c)}\triangleq \sum\limits _{n=1}^{N}\left|h_{n}\right|\left|g_{n,e}\right|\cos\left(\phi_{n}\right),\label{hfec}
\end{equation}
\begin{equation}
h_{f,e}^{(s)}\triangleq \sum\limits _{n=1}^{N}\left|h_{n}\right|\left|g_{n,e}\right|\sin\left(\phi_{n}\right),\label{hfes}
\end{equation}
with  $j=\sqrt{-1}$, $\phi_{n}=-\angle g_{n,e}-\angle g_{n,\ell}$ are uniformly distributed in $[0,2\pi]$.

For the sake of simplicity, let $Y$ denote either $\left|h_{n}\right|$
or $\left|g_{n,x}\right|$, supposed $\alpha-\mu$ distributed.
Its PDF and CDF can be expressed as 
\begin{equation}
f_{Y}\left(z\right)=\frac{\alpha\mu^{\mu}}{\overline{h}^{\alpha\mu}\Gamma\left(\mu\right)}z^{\alpha\mu-1}\exp\left(-\frac{\mu}{\overline{h}^{\alpha}}z^{\alpha}\right);z\ge0,
\end{equation}\vspace{-0.25cm}
\begin{equation}
F_{Y}\left(z\right)=\gamma\left(\mu,\frac{\mu}{\overline{h}^{\alpha}}z^{\alpha}\right),
\end{equation}
with $\gamma\left(\cdot,\cdot\right)$ stands to the lower incomplete Gamma function and $\overline{h}=\sqrt[\alpha]{\mathbb{E}\left[Y^{\alpha}\right]}$. On the other hand, the PDF and CDF of $h_{p,x}^{2}$ can be approximated as \cite{4267802}
\begin{align}
f_{h_{p,x}^{2}}\left(z\right) & =\frac{\varphi}{2A_{0}^{\varphi}}z^{\frac{\varphi}{2}-1};0\leq z\leq A_{0}^{2},\label{eq:PDFhpx2}
\end{align}
and\vspace{-0.25cm}
\begin{equation}
F_{h_{p,x}^{2}}\left(z\right)  =\left\{ \begin{array}{c}
\frac{z^{\frac{\varphi}{2}}}{A_{0}^{\varphi}};0\leq z\leq A_{0}^{2}\\
1,z>A_{0}^{2}
\end{array}\right..\label{CDFhp2}
\end{equation}
with $A_{0}=\text{erf}^{2}\left(\upsilon\right)$ is the fraction of the collected power at \textcolor{black}{the} receivers when there is
no pointing error ($z=$0). Details on $A_0$ computation can be found in \cite{4267802}. Of note, the first moment of $h_{p,x}^{2}$ can be evaluated as 
\begin{equation}
    \mathbb{E}\left[h_{p,x}^{2}\right]	=\frac{\varphi}{2A_{0}^{\varphi}}\int_{0}^{A_{0}^{2}}z^{\frac{\varphi}{2}}dz
	=\frac{\varphi A_{0}^{2}}{\varphi+2}.\label{momh_px2}
\end{equation}\vspace{-0.45cm}
\section{Statistical Functions of $\left|h_{x}\right|^{2}$}
In this section, the CDF and PDF of $\left|h_{\ell}\right|^{2}$ and $\left|h_{x}\right|^{2}$, respectively, are provided, in the case of large RIS, to determine that of $\gamma_x$ relying on \eqref{CDFgamma} as\vspace{-0.25cm}
\begin{equation}
    F_{\gamma_{x}}(x)=F_{\left|h_{x}\right|^{2}}\left(\frac{z\sigma_{x}^{2}}{\mathcal{P}h_{l,x}^{2}}\right);x=(\ell,e),\label{CDFhgammax}
\end{equation}\vspace{-0.3cm}
and 
\begin{equation}
   f_{\gamma_{x}}(x)=\frac{\sigma_{e}^{2}}{\mathcal{P}h_{l,x}^{2}}f_{h_{x}^{2}}\left(\frac{z\sigma_{x}^{2}}{\mathcal{P}h_{l,x}^{2}};\right);x=(\ell,e),\label{PDFgammax}
\end{equation} \vspace{-0.5cm}
\subsection{CDF of $\left|h_{\ell}\right|^{2}$}
\begin{proposition}
    For substantially larger $N$ values (i.e., $N>10$), the CDF of $h_{\ell}^{2}$ can be tightly approximated as \begin{align}
F_{h_{\ell }^{2}}\left( z\right) & \approx 1-Q\left( \frac{\frac{\sqrt{z}}{%
A_{0}}-\sqrt{\mathcal{G}_{N}}}{\sqrt{\Psi _{N}}}\right)   \notag \\
& +\mathcal{B}_{N}\sum_{k=0}^{\infty}\underset{\triangleq\mathcal{T}_{k}(z)}{\underbrace{\frac{\left(\frac{2\mathcal{G}_{N}}{\Psi_{N}}\right)^{\frac{k}{2}}z^{\frac{\varphi}{2}}\Gamma\left(\frac{k-\varphi+1}{2},
\frac{z}{2\Psi_{N}A_{0}^{2}}\right)}{k!}}},  \label{eq:CDFhl}
\end{align}\vspace{-0.25cm}
with
    \begin{equation}   \Psi_{N}=N\left(\mathcal{H}_{2}^{2}-\mathcal{H}_{1}^{4}\right),\label{PsiN}
\end{equation}\vspace{-0.25cm} 
\begin{equation}
\mathcal{G}_{N}=\left(N\mathcal{H}_{1}^{2}\right)^{2},\label{GN}
\end{equation}\vspace{-0.25cm}
\begin{equation}   \mathcal{B}_{N}=\frac{\left(2\Psi_{N}A_{0}^{2}\right)^{-\frac{\varphi}{2}}}{2\sqrt{\pi}\exp\left(\frac{\mathcal{G}_{N}}{2\Psi_{N}}\right)},\label{BN}
\end{equation}\vspace{-0.05cm}
\begin{equation}   \mathbb{E}\left[h_{\ell}^{2}\right]=\frac{N\varphi A_{0}^{2}\left(\mathcal{H}_{2}^{2}+\left(N-1\right)\mathcal{H}_{1}^{4}\right)}{\varphi+2},\label{exphl2}
\end{equation} \vspace{-0.15cm} 
\begin{equation}
\mathcal{H}_{i} \textcolor{black}{\triangleq}  \frac{\Gamma\left(\mu+\frac{i}{\alpha}\right)}{\Gamma\left(\mu\right)\mu^{\frac{i}{\alpha}}}\overline{h}^{i},i\ge1;\label{Hi}
\end{equation} $\Gamma\left(\cdot,\cdot\right)$ refers to the upper incomplete Gamma function. and $Q(\cdot)$ account for the $Q$-function.%
\end{proposition}
\begin{proof}
    The proof is provided in Appendix A.
    \end{proof}
    \begin{corollary}
        The infinite series in \eqref{eq:CDFhl} is absolutely convergent and the error of the truncated series is upper bounded by 
        \begin{equation}
            \epsilon_{K_{\max}}(z)\leq\left\{ \begin{array}{c}
\mathcal{W}_{K_{\max}},z\geq A_{0}^{2}\mathcal{G}_{N}\\
2\mathcal{W}_{K_{\max}},z<A_{0}^{2}\mathcal{G}_{N}
\end{array}\right.,\label{epsilon}
        \end{equation}\vspace{-0.35cm}
        where 
        \begin{equation}           \mathcal{W}_{K_{\max}}=\left(\frac{\mathcal{G}_{N}}{\Psi_{N}}\right)^{K_{\max}+1}\sum_{i=0}^{K_{\max}+1}\frac{\left(\frac{\Psi_{N}}{2\mathcal{G}_{N}}\right)^{\frac{i}{2}}}{\Gamma\left(\frac{i}{2}+1\right)\left(K_{\max}+1-i\right)!}.\label{Wkmax}
        \end{equation}
    \end{corollary}
\begin{proof}
    The proof is provided in Appendix B.
    \end{proof}\vspace{-0.25cm}
    \begin{remark}
     The above bound is irrespective of $z$. Further, as $i/2$ in the finite series is bounded by $(K_{\max}+1)/2$, it turns that $\mathcal{W}_{K_{\max}}$ is proportional to $\mathcal{G}_{N}/\Psi_{N}=N/(\mathcal{H}_2^2/\mathcal{H}_1^4-1)$. Thus the error is increasing with $N$. Further, substituting \eqref{Hi} into \eqref{GN}-\eqref{PsiN}, it can be seen that the error is proportional to the following function\vspace{-0.2cm}
     \begin{equation}      \Phi\left(\alpha,\mu\right)=\frac{\Gamma^{2}\left(\mu+\frac{1}{\alpha}\right)}{\Gamma\left(\mu\right)\Gamma\left(\mu+\frac{2}{\alpha}\right)}.
     \end{equation}
     Its two partial derivative with respect to $\alpha$ and $\mu$ can be evaluated as\vspace{-0.15cm}
     \begin{equation}      \frac{\partial\Phi\left(\alpha,\mu\right)}{\partial\alpha}=\frac{2\Gamma\left(\mu+\frac{1}{\alpha}\right)\left[\psi\left(\mu+\frac{2}{\alpha}\right)-\psi\left(\mu+\frac{1}{\alpha}\right)\right]}{\alpha^{2}\Gamma\left(\mu\right)\Gamma\left(\mu+\frac{2}{\alpha}\right)},
     \end{equation}\vspace{-0.25cm}
     \begin{equation}       \frac{\partial\Phi\left(\alpha,\mu\right)}{\partial\mu}	=\frac{\Gamma^{2}\left(\mu+\frac{1}{\alpha}\right)\left[2\left(1+\frac{1}{\alpha^{2}}\right)\psi\left(\mu+\frac{2}{\alpha}\right)-\psi\left(\mu\right)\right]}{\Gamma\left(\mu\right)\Gamma\left(\mu+\frac{2}{\alpha}\right)},
     \end{equation}
     where $\psi(\cdot)$ represents the Digamma function, which is increasing. 
     As a result, both derivatives are positive and the error rises with the increase of either $\alpha$ or $\mu$, requiring a greater $K_{\max}$ to lower it.
    \end{remark}  \vspace{-0.5cm}
\subsection{PDF of $\left|h_{e}\right|^{2}$}
\begin{proposition}
For somewhat higher N numbers, the PDF of $h_{\ell}^{2}$ may be accurately
approximated as\vspace{-0.15cm}
\begin{align}
f_{\left\vert h_{e}\right\vert ^{2}}\left(z\right)\approx\left\{ \begin{array}{c}
\frac{\varphi z^{\frac{\varphi}{2}-1}}{2\left(N\mathcal{H}_{2}^{2}A_{0}^{2}\right)^{\frac{\varphi}{2}}}\Gamma\left(-\frac{\varphi}{2}+1,\frac{z}{N\mathcal{H}_{2}^{2}A_{0}^{2}}\right);z>0\\
\frac{\varphi}{N\mathcal{H}_{2}^{2}A_{0}^{2}\left(\varphi-2\right)};z=0
\end{array}\right.,\label{pdfhe222}
\end{align}
with\vspace{-0.15cm}
\begin{equation} \mathbb{E}\left[\left\vert h_{e}\right\vert ^{2}\right]=\frac{\varphi N \mathcal{H}_{2}^{2}A_{0}^{2}}{\varphi+2}.\label{exphe2}
\end{equation}
\end{proposition}
\begin{proof}
    The proof is provided in Appendix C.
\end{proof}
\begin{remark}
    Using \eqref{exphl2} and \eqref{exphe2} along with \eqref{eq:SNR}, the average SNRs at node $x$ can be expressed as
\begin{equation}
\textcolor{black}{\overline{\gamma}_{\ell}}=\frac{N\varphi A_{0}^{2}\mathcal{P}h_{l,\ell}^{2}\left(\mathcal{H}_{2}^{2}+\left(N-1\right)\mathcal{H}_{1}^{4}\right)}{\sigma_{\ell}^{2}\left(\varphi+2\right)},\label{gambl}
\end{equation}
and\vspace{-0.15cm}
\begin{equation}
\textcolor{black}{\overline{\gamma}_{e}}=\frac{N\varphi A_{0}^{2}\mathcal{P}h_{l,e}^{2}\mathcal{H}_{2}^{2}}{\sigma_{e}^{2}\left(\varphi+2\right)}.\label{gambe}
\end{equation}
\end{remark}
Thus, for equal $\sigma_{x}^{2}$ and path gain, \textcolor{black}{$\overline{\gamma}_{\ell}>\overline{\gamma}_{e}$}.
\section{SOP of the Considered System}
\begin{proposition}
    For a given secrecy rate $R_s$, the SOP and IP of the considered system can be approximated as
    \begin{equation}      \mathcal{P}_{\text{sop}}(R_s)\approx 1-\frac{1}{\mathcal{M}}\left[\mathcal{E}\left(\mathcal{D}_{2}^{(\mathcal{L})}\right)-\mathcal{B}_{N}\sum_{k=0}^{\infty}\mathcal{E}\left(\mathcal{D}_{1,k}^{(\mathcal{L})}\right)\right],  \label{SOP}
    \end{equation}
    where\vspace{-0.15cm} 
    \begin{equation}        \mathcal{E}\left(\mathcal{D}_{\cdot}^{(\mathcal{L})}\right)=\left\{ \begin{array}{c}
\frac{1}{48\mathcal{L}\mathcal{S}}\sum_{s=1}^{\mathcal{S}+1}c_{s}\mathcal{D}_{\cdot}^{(\mathcal{L})}\left(\frac{s}{\mathcal{L}\mathcal{S}}\right),R_{s}>0\\
\sum_{s=0}^{\mathcal{S}}w_{s}\mathcal{D}_{\cdot}^{(0)}\left(y_{s}\right),R_{s}=0
\end{array}\right.,\label{Espsilonfunct}
    \end{equation}\vspace{-0.1cm} 
    \begin{equation}       \mathcal{D}_{1,k}^{(\mathcal{L})}\left(z\right)=\left\{ \begin{array}{c}
z^{-2}\mathcal{Y}_{k}^{(\mathcal{L})}\left(\frac{1}{z}\right);\mathcal{L}>0,z<\frac{1}{\mathcal{L}}\\
\mathcal{Y}_{k}^{(0)}\left(z\right);\mathcal{L}=0
\end{array}\right.,\label{D1k}
    \end{equation}\vspace{-0.1cm}    
    \begin{equation}     \mathcal{D}_{2}^{(\mathcal{L})}\left(z\right)=\left\{ \begin{array}{c}
z^{-2}\mathcal{Z}^{(\mathcal{L})}\left(\frac{1}{z}\right);\mathcal{L}>0,z<\frac{1}{\mathcal{L}}\\
\textcolor{black}{\mathcal{Z}^{(0)}}\left(z\right);\mathcal{L}=0
\end{array}\right.,\label{D2}
    \end{equation}
    \begin{equation}
        \mathcal{Y}_{k}^{(\mathcal{L})}(z)=\mathcal{T}_{k}\left(z\right)f_{\left|h_{e}\right|^{2}}\left(\frac{z-\mathcal{L}}{\mathcal{M}}\right),
    \end{equation}
    \begin{equation}
        \mathcal{Z}^{(\mathcal{L})}(z)=Q\left(\frac{\frac{\sqrt{z}}{A_{0}}-\sqrt{\mathcal{G}_{N}}}{\sqrt{\Psi_{N}}}\right)f_{\left|h_{e}\right|^{2}}\left(\frac{z-\mathcal{L}}{\mathcal{M}}\right),
    \end{equation}\vspace{-0.15cm}
with $\mathcal{S}$ as positive integer above $8$,  $\mathcal{M}=\frac{2^{R_{s}}h_{l,e}^{2}\sigma_{\ell}^{2}}{h_{l,\ell}^{2}\sigma_{e}^{2}}$, $\mathcal{L}=\frac{\sigma_{\ell}^{2}\left(2^{R_{s}}-1\right)}{\mathcal{P}h_{l,\ell}^{2}}$, \textcolor{black}{$c_s$} are provided in Table \ref{C-coefficients}, $y_{s}$ is the $s$-th root of the Laguerre polynomial $L_{\mathcal{S}}$
of order $\mathcal{S}$, \textcolor{black}{$\mathcal{T}_{k}(\cdot)$ is being defined in \eqref{eq:CDFhl}}, and\vspace{-.15cm}
\begin{equation}
w_{s}=\frac{y_{s}}{\left[\left(\mathcal{S}+1\right)L_{\mathcal{S}+1}\left(y_{s}\right)\right]^{2}}.
\end{equation}\vspace{-0.35cm}
\begin{center}
\begin{table}[H]
\caption{\textcolor{black}{Coefficients $c_s$.}}
\color{black}
\begin{tabular}{c|c|c|c|c}
\hline 
$c_{\textcolor{black}{1}}=c_{\textcolor{black}{\mathcal{S}+1}}$ & $c_{\textcolor{black}{2}}=c_{\textcolor{black}{\mathcal{S}}}$ & $c_{\textcolor{black}{3}}=c_{\textcolor{black}{\mathcal{S}-1}}$ & $c_{\textcolor{black}{4}}=c_{\textcolor{black}{\mathcal{S}-2}}$ & $\{c_{s}\}_{\textcolor{black}{5}\le s\le \textcolor{black}{\mathcal{S}-3}}$\tabularnewline
\hline 
\hline 
$17$ & $59$ & $43$ & $49$ & $48$\tabularnewline
\hline 
\end{tabular}\vspace{-0.7cm}
    	\label{C-coefficients}
			\end{table}
		\par\end{center}
\end{proposition}
\begin{proof}
       The proof is provided in Appendix D. 
\end{proof}
\color{black}
\begin{remark} \label{rk1}
\begin{enumerate}[label=(\Alph*)]
\item It's worth mentioning that $\mathcal{M}$ is constant for equal $\sigma_x$ and distance $d_{r,x}$. Thus, the IP becomes independent of these parameters as $\mathcal{L}=0$.
\item Using Eqs. \eqref{gambl}-\eqref{gambe} along with the approximation of $\left\vert h_{x}\right\vert$ for $N\geq 10$, one can write\vspace{-.25cm}
\begin{align}
\frac{\gamma_{\ell}}{\gamma_{e}} & =\frac{\overline{\gamma}_{\ell}X}{\overline{\gamma}_{e}Y}\nonumber \\
 & =\mathcal{J}(N)\times\frac{X}{Y},\label{ratiogammas}
\end{align}
where $X$ and $Y$ are two RVs of unit mean and irrespective of $N$, and\vspace{-.25cm}
\begin{equation}
\mathcal{J}(N)=\left(1+\left(N-1\right)\frac{\mathcal{H}_{1}^{4}}{\mathcal{H}_{2}^{2}}\right)\frac{h_{l,\ell}^{2}}{h_{l,e}^{2}}.
\end{equation}\vspace{-.25cm}
Further,\vspace{-.25cm}
\begin{align}
\mathcal{P}_{\text{sop}}(R_{s}) & \triangleq\Pr\left(\log_{2}\left(\frac{1+\gamma_{\ell}}{1+\gamma_{e}}\right)\leq R_{s}\right)\nonumber \\
 & \overset{(a)}{\approx}\Pr\left(\log_{2}\left(\frac{\gamma_{\ell}}{\gamma_{e}}\right)\leq R_{s}\right)\nonumber \\
 & \overset{(b)}{=}\Pr\left(\frac{X}{Y}\leq\frac{2^{R_{s}}}{\mathcal{J}(N)}\right),
\end{align}
where step (a) holds for significantly high SNR values, while step (b) follows by utilizing \eqref{ratiogammas}. Now, because $\mathcal{J}(N)$
is increasing on $N$, it follows that the system becomes more secure
as $N$ rises.
\end{enumerate}
\end{remark}
\begin{proposition}
    The asymptotic expression for the SOP can be evaluated as
    \begin{equation}
        \mathcal{P}_{\text{sop}}(R_{s})=\sum_{s=0}^{\mathcal{S}}w_{s}\mathcal{K}\left(y_{s}\right)\exp\left(y_{s}\right)-\mathcal{L}\times \mathcal{O}_{N},\label{SOPasympt0}
    \end{equation}\vspace{-.25cm}
    where \vspace{-.25cm}
    \begin{equation}
        \mathcal{O}_{N}=\frac{\varphi Q\left(\sqrt{\frac{\mathcal{G}_{N}}{\Psi_{N}}}\right)}{\mathcal{M}N\mathcal{H}_{2}^{2}A_{0}^{2}\left(\varphi-2\right)},\label{ON}
    \end{equation}\vspace{-.25cm}
    and \vspace{-.1cm}
        \begin{equation}     \mathcal{K}\left(z\right)=F_{h_{\ell}^{2}}\left(\mathcal{M}z\right)\times\left(\begin{array}{c}
f_{\left|h_{e}\right|^{2}}\left(z\right)\left(1-\mathcal{L}\left(\frac{\varphi}{2}-1\right)z^{-1}\right)\\
+\mathcal{L}\left[\frac{\varphi z^{-1}}{2N\mathcal{H}_{2}^{2}A_{0}^{2}}\exp\left(-\frac{z}{N\mathcal{H}_{2}^{2}A_{0}^{2}}\right)\right]
\end{array}\right).\label{Kz}
        \end{equation}
\end{proposition}
\begin{proof}
       The proof is provided in Appendix E. 
\end{proof}
\begin{remark}
\begin{enumerate}[label=(\Alph*)]
    \item When $R_s=0$, $\mathcal{L}=0$. As a result, the SOP tends to the IP;
    \item the secrecy diversity order equals 1 and $\mathcal{O}_N$ is decreasing on $N$.
    \end{enumerate}
\end{remark}
\color{black}
\section{Results and Numerical Results}
This section investigates the statistical functions of legitimate and wiretap channels, as well as the SOP behavior for the system \textcolor{black}{under} consideration. The simulation was conducted using the Monte Carlo technique, generating $7\times10^{6}$ random $\alpha-\mu$ samples with the following set of parameters: $G_{r,\ell}=G_{r,t}=G_{t}=40\mathrm{\ dB},f=300\mathrm{\ GHz},d_{\mathrm{1}}=5\mathrm{\ m},d_{\mathrm{r,\ell}}=20\mathrm{\ m},\mathcal{P}/\sigma^2=60\mathrm{\ dB},\kappa_{a}(f)=3.18\times10^{-4}\mathrm{\ m}^{-1},\varphi=25.7404,$
and $A_{0}=0.054$.
Fig. \ref{CDFhl22} shows the CDF of $h_{\ell}^{2}$ for $\overline{h}=1.5$
and different values of $N,\alpha,$ and $\mu$, plotted using \eqref{eq:CDFhl}.
The infinite series was shortened to $10$ and $15$ initial
terms, respectively. The analytical and simulation curves are identical
for each configuration, demonstrating the exceptional accuracy of
the obtained analytical results. Further, increasing $N,\alpha,$
and $\mu$ leads to an increase in $K_{\max}$ to ensure a tiny error,
which supports \textbf{Remark 1.}

\begin{figure*}[ht]

\begin{minipage}{0.3\linewidth}
\centering  
\hspace*{-.2cm}\includegraphics[scale=.1475]{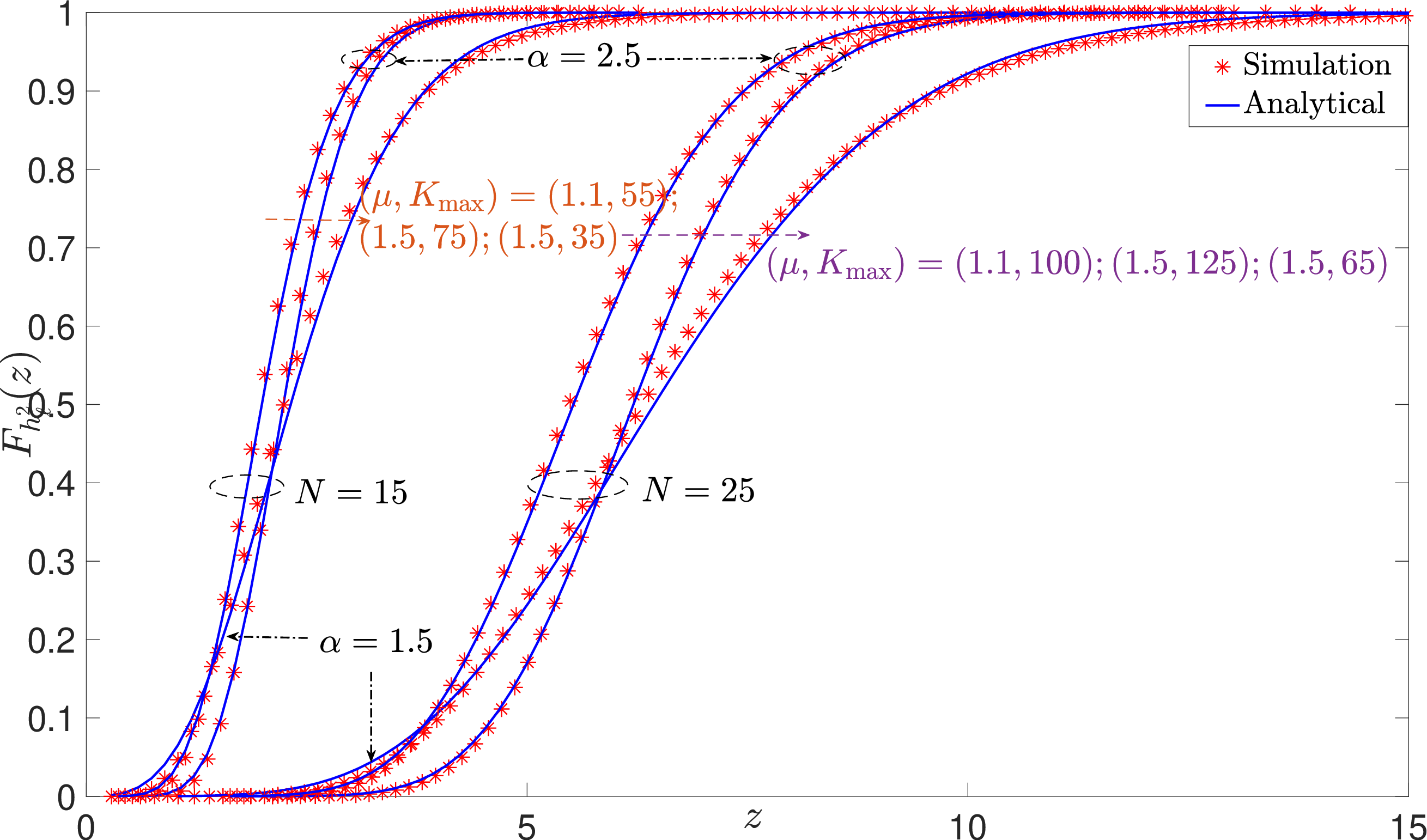} 	
\caption{Simulated and approximate $F_{h_{\ell}^{2}}\left(z\right)$ for $\overline{h}=1.5$ and  various value of $\alpha,\mu,$ and $N$.}
\label{CDFhl22}
\end{minipage}%
\hfill
\begin{minipage}{0.3\linewidth}
\centering  
\hspace*{-.2cm}\includegraphics[scale=.1475]{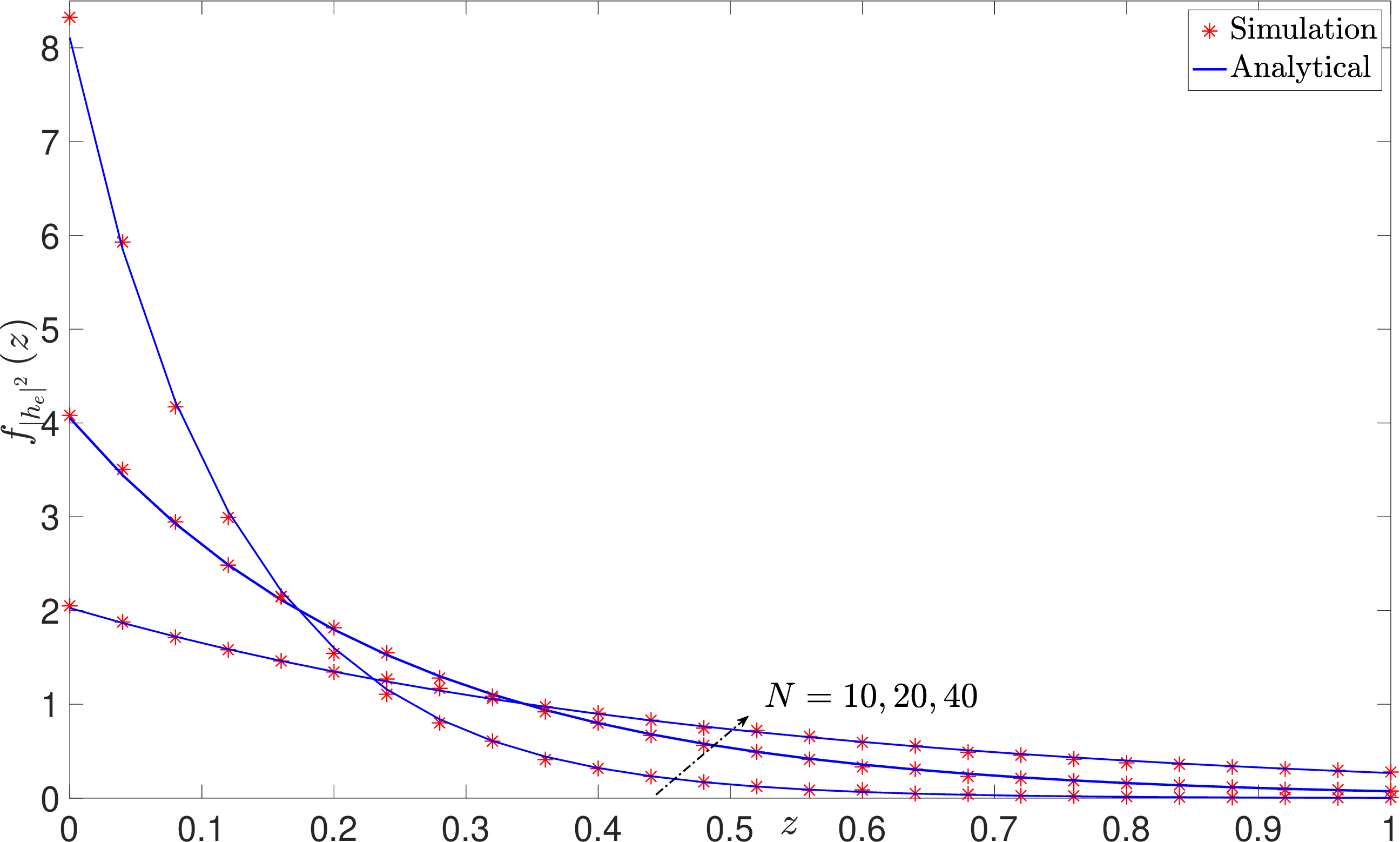} 	
\caption{Simulated and approximate $f_{\left\vert h_{e}\right\vert ^{2}}\left(z\right)$ for $\alpha=2.5,\mu=1.5,\overline{h}=1.5$ and different value of $N$.}
\label{PDFhe2}
\end{minipage}
\hfill
\begin{minipage}{0.3\linewidth}
\centering  
\vspace*{-.2cm}
\hspace*{-.2cm}\includegraphics[scale=.1475]{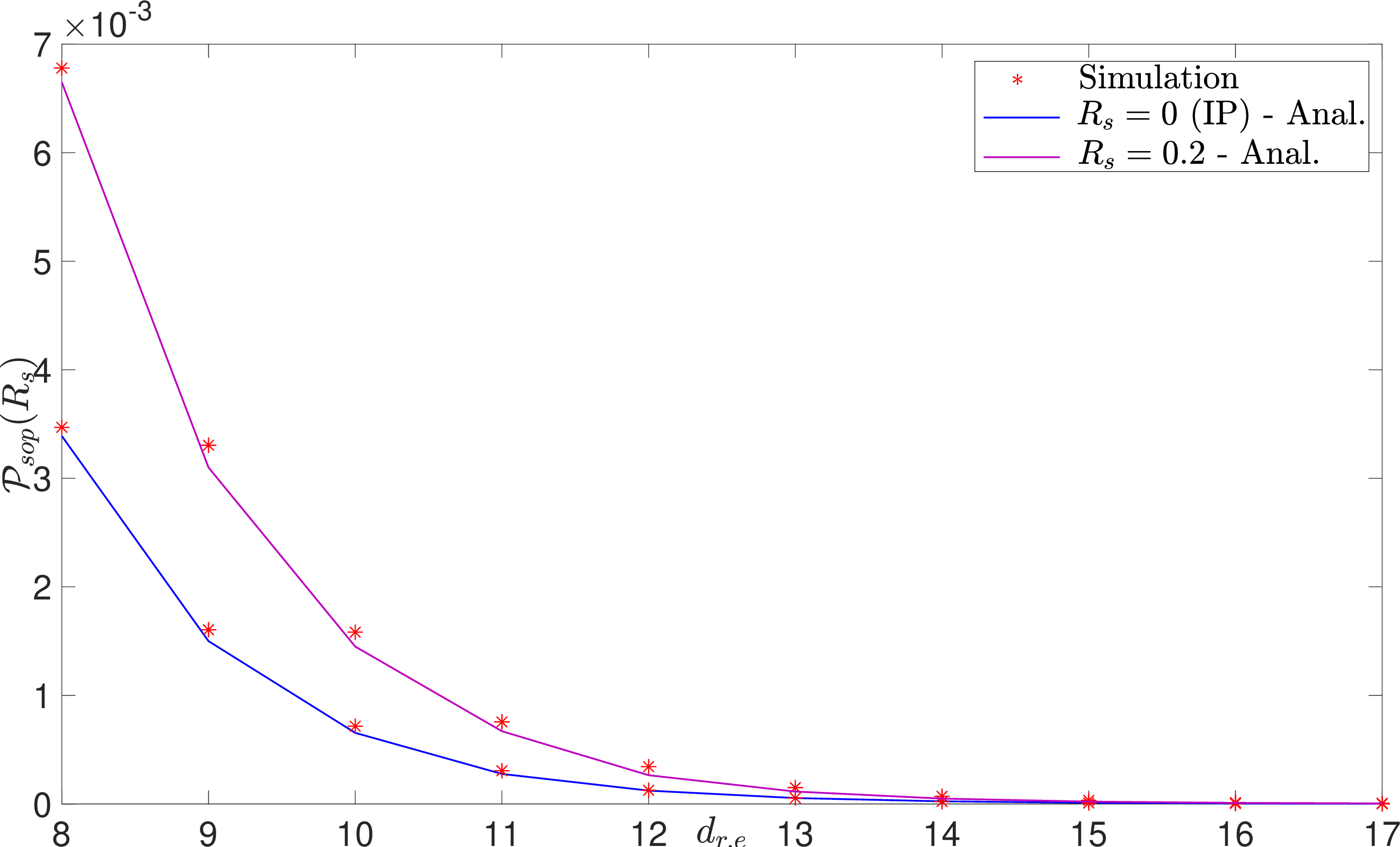}
\caption{\textcolor{black}{Simulated and approximate $\mathcal{P}_{sop}$ vs $d_{r,e}$ for $\overline{h}=1,N=60,$ and $\mathcal{P}/\sigma_{x}^{2}=60$ dB.}}
		\label{IPSOPFig} 
\end{minipage}%
\end{figure*}
\begin{figure}
    \centering
   \includegraphics[scale=.1475]{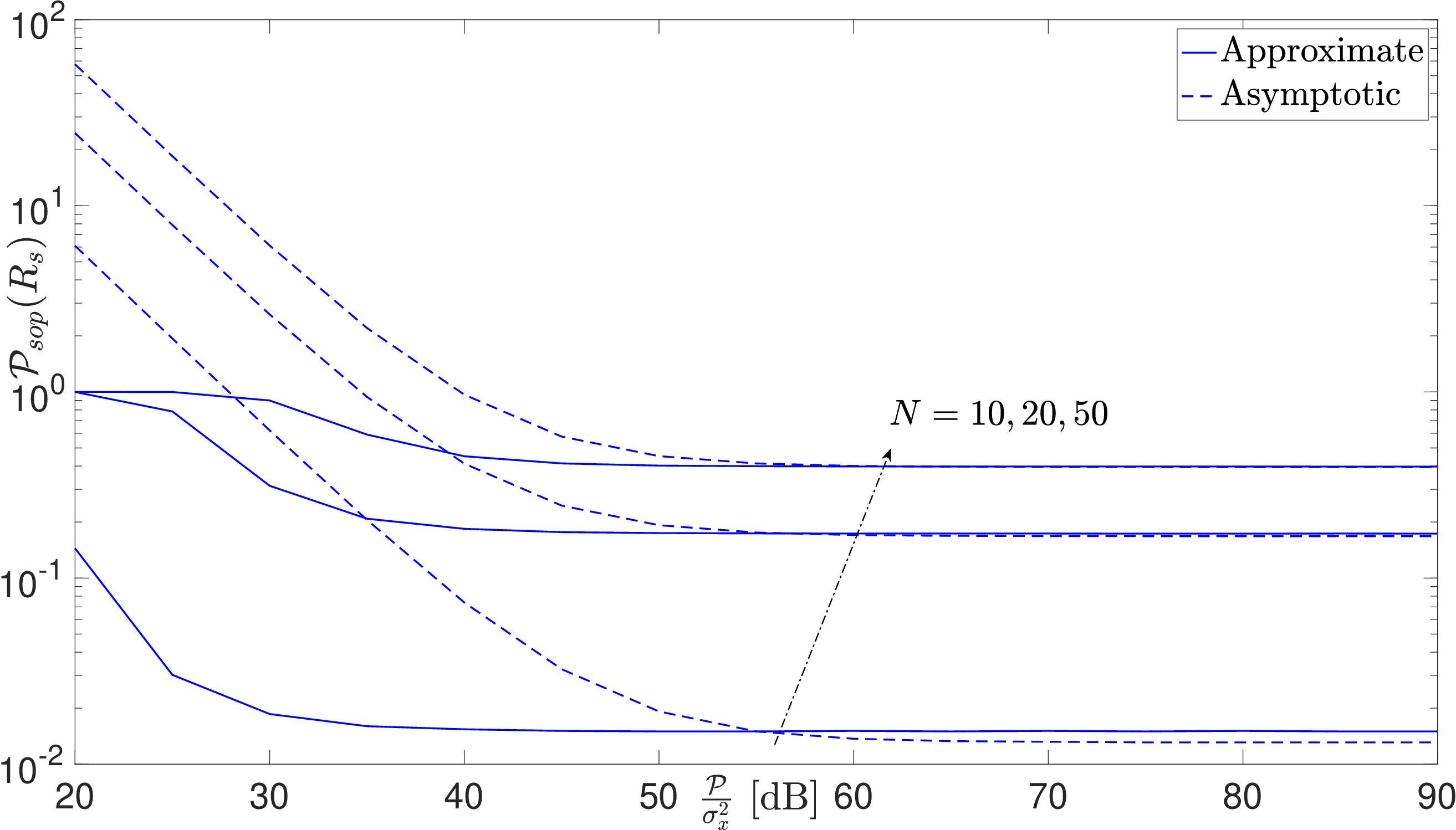} 	
\caption{\textcolor{black}{Aymptotic SOP for   $\alpha=1.7,\mu=1.1,R_s=0.2,d_{r,e}=8$ m, and various values of $N$.}}
    \label{fig:enter-label}
\end{figure}
Fig. \ref{PDFhe2} depicts the PDF of $\left\vert h_{e}\right\vert^{2}$ for several values of $N$. Similarly, the simulation curves validate our finding and prove the approximation tightness for $N\ge 10$. Fig. \ref{IPSOPFig} presents the IP $(R_s=0)$ and SOP $(R_s=0.2)$ for $\alpha=1.7,\mu=1.1,\textcolor{black}{K_{\max}=0},d_1=5\text{\ m}$, and $d_{r,\ell}=25\text{\ m}$. Analytical curves almost match the simulated ones. Further,  the infinite series in \eqref{SOP} is neglected to the first term\textcolor{black}{, justifying} the choice of $K_{\max}=0$ and that the SOP is proportional to $R_s$. \textcolor{black}{Besides, the greater $d_{r,e}$, the smaller $h_{l,e},\mathcal{M},$ and $\gamma_{e}$. Thus, the SOP is decreasing on $d_{r,e}$ according to Eq. \eqref{SOP0}. Further, for high SNR values, $\mathcal{L}$ tends to $0$ and the SOP approaches IP, corroborating \textbf{Remark 4}.  The asymptotic SOP presented in Fig. 4, corroborates the accuracy of Eq. \eqref{SOPasympt0} and \textbf{Remark 3.B.}}\vspace{-0.35cm}
\section{\color{black}Conclusion}	
\color{black}
A RIS-assisted THz WCS subject to $\alpha-\mu$ with pointing errors has been considered in the presence of a wiretapper. RIS phase shifts 
were tuned to maximize the SNR at the legitimate receiver. Simple closed-form approximation formulas for the CDF and PDF of the channel gains were obtained, from which the SOP and IP were determined for THz communication systems subject to $\alpha-\mu$ fading and PEs. Finally, the simulation findings corroborate the theoretical analysis. \textcolor{black}{Future work may consider imperfect CSI and hardward impairments, though they pose analytical challenges under $\alpha$–$\mu$ fading with PEs.}
\color{black}\vspace{-0.3cm}
\section*{Appendix A\protect \protect \protect \\
		Proof of Proposition 1}
 Let's first find the distribution of $h_{f,\ell}$ given in \eqref{eq:hfl} to deduce that of $h_{\ell}$ using \eqref{hx} along with \eqref{eq:PDFhpx2}.\vspace{-0.3cm}
\subsection*{Approximate CDF of $h_{f,\ell}^{2}$}
 Given the independence of the RVs $\left|h_{n}\right|$ and $\left|g_{n,\ell}\right|$, it follows using the central limit theorem (CLT) that $h_{f,\ell}$ defined in \eqref{eq:hfl} is a normal distribution with the two first moments are expressed as\vspace{-0.25cm} 
 \begin{equation}  \mathbb{E}\left[h_{f,\ell}\right]=N\mathcal{H}_{1}^{2},\label{mom1}
 \end{equation}\vspace{-0.35cm}
 \begin{equation}  \mathbb{E}\left[h_{f,\ell}^{2}\right]=N\mathcal{H}_{2}^{2}+N\left(N-1\right)\mathcal{H}_{1}^{4},\label{mom2}
 \end{equation}
 where $\mathcal{H}_{i}$ denote the $i$th moment of $\alpha-\mu$ RV, i.e., $\mathcal{H}_{i}\triangleq\mathbb{E}\left[Y_{n}^{i}\right]$, and is given in \eqref{Hi}. Given that $h_{f,\ell}$ is a positive RV, its PDF and CDF can be approximated using \eqref{mom1}-\eqref{mom2} as\vspace{-0.15cm} 
 \begin{equation}
f_{h_{f,\ell}}\left(z\right)	\approx \frac{1}{\sqrt{2\pi\Psi_{N}}}\exp\left(-\frac{\left(z-\sqrt{\mathcal{G}_{N}}\right)^{2}}{2\Psi_{N}}\right);z\ge 0,
\end{equation}\vspace{-0.35cm}
\begin{equation}
F_{h_{f,\ell}}\left(z\right)\approx 1-Q\left(\frac{z-\sqrt{\mathcal{G}_{N}}}{\sqrt{\Psi_N}}\right);z\ge0,\label{CDFhflll}\end{equation}\vspace{-0.1cm}
where $\Psi_{N}$ and $\mathcal{G}_{N}$ are being defined in \eqref{PsiN} and \eqref{GN}.\vspace{-0.1cm}
\subsection*{Approximate CDF of $h_{\ell}^{2}$}
 Relying on \eqref{hx}, the CDF of $h_{\ell}^{2}$ can be expressed as \vspace{-0.25cm}
 \begin{align}   F_{h_{\ell}^{2}}\left(z\right)&\triangleq\text{Pr}\left(h_{p,\ell}<\frac{\sqrt{z}}{h_{f,\ell}}\right)\nonumber\\&=\int_{0}^{\infty}F_{h_{p,\ell}}\left(\frac{\sqrt{z}}{y}\right)f_{h_{f,\ell}}\left(y\right)dy\nonumber\\&\overset{(a)}{=}\int_{\frac{\sqrt{z}}{A_{0}}}^{\infty}\left(\frac{\sqrt{z}}{A_{0}y}\right)^{\varphi}f_{h_{f,\ell}}\left(y\right)dy+\int_{0}^{\frac{\sqrt{z}}{A_{0}}}f_{h_{f,\ell}}\left(y\right)dy\nonumber\\&\overset{(b)}{\approx}\underset{\triangleq\mathcal{C}(z)}{\underbrace{\frac{1}{\sqrt{2\pi\Psi_{N}}}\int_{\frac{\sqrt{z}}{A_{0}}}^{\infty}\exp\left(-\frac{\left(y-\sqrt{\mathcal{G}_{N}}\right)^{2}}{2\Psi_N}\right)\left(\frac{\sqrt{z}}{A_{0}y}\right)^{\varphi}dy}}\nonumber\\&+1-Q\left(\frac{\frac{\sqrt{z}}{A_{0}}-\sqrt{\mathcal{G}_{N}}}{\sqrt{\Psi_N}}\right),\label{CDFhl222}
 \end{align}
where steps (a) and (b) hold using \eqref{CDFhp2} and \eqref{CDFhflll}.
Further, using the Binomial Theorem along with the Maclaurin series and employing \cite[Eq. (3.381.9)]{gradshteyn2007}, one gets \vspace{-0.25cm}
\begin{align}
    \mathcal{C}(z)	&=\frac{\exp\left(-\frac{\mathcal{G}_{N}}{2\Psi_{N}}\right)}{\sqrt{2\pi\Psi_{N}}A_{0}^{\varphi}}z^{\frac{\varphi}{2}}\sum_{k=0}^{\infty}\frac{1}{k!}\left(\frac{\sqrt{\mathcal{G}_{N}}}{\Psi_{N}}\right)^{k}\nonumber\\&
	\times\underset{=2\left(2\Psi_{N}\right)^{-\frac{k-\varphi+1}{2}}\Gamma\left(\frac{k-\varphi+1}{2},\frac{z}{2\Psi_{N}A_{0}^{2}}\right)}{\underbrace{\int_{\frac{\sqrt{z}}{A_{0}}}^{\infty}y^{k-\varphi}\exp\left(-\frac{y^{2}}{2\Psi_{N}}\right)dy}}.\label{Tk}
\end{align}
This concludes the proof of Proposition 1.\vspace{-0.5cm}
\section*{Appendix B\protect \protect \protect \\
		Proof of Corollary 1}
  To prove the absolute convergence of the infinite series in \eqref{eq:CDFhl}, it is sufficient to check that  $\mathcal{T}_{k+1}(z)/\mathcal{T}_{k}(z)$ falls bellow $1$ above certain threshold of $k$.
  It can be verified that\vspace{-0.15cm}
  \begin{equation}
      \frac{\mathcal{T}_{k+1}(z)}{\mathcal{T}_{k}(z)}=\frac{1}{k+1}\left(\frac{2\mathcal{G}_{N}}{\Psi_{N}}\right)^{\frac{1}{2}}\underset{\triangleq\mathcal{U}_{k}(z)}{\underbrace{\frac{\Gamma\left(\frac{k-\varphi+1}{2},\frac{z}{2\Psi_{N}A_{0}^{2}}\right)}{\Gamma\left(\frac{k+1-\varphi+1}{2},\frac{z}{2\Psi_{N}A_{0}^{2}}\right)}}}.
  \end{equation}\vspace{-0.1cm}
  Further, using the two inequalities $\Gamma\left(a,x\right)\geq x^{a-1}e^{-x}$ and $\Gamma\left(a+\frac{1}{2},x\right)\leq \frac{x^{a+\frac{1}{2}}e^{-x}}{\left(x-\left(a+\frac{1}{2}-1\right)\ln x\right)}$ for significant higher values of $a$ \cite{Borwein2009UniformBF}, one gets\vspace{-0.15cm}
  \begin{equation}
      \mathcal{U}_{k}(z)\leq\frac{z_{N}}{\left|z_{N}-\left(\frac{k-\varphi}{2}\right)\ln z_{N}\right|}.
  \end{equation}\vspace{-.08cm}
 Subsequently $\underset{k\rightarrow\infty}{\text{lim}}\frac{\mathcal{T}_{k+1}(z)}{\mathcal{T}_{k}(z)}=0$.
  On the other hand, the error of Maclaurin truncation can be expressed by evaluating the $(K_{\max}+1)$th derivative of $\exp\left(y\sqrt{\mathcal{G}_{N}}/\Psi_{N}\right)$ as\vspace{-0.15cm}
  \begin{align}
      \epsilon_{K_{\max}}(z)&=\frac{\exp\left(-\frac{\mathcal{G}_{N}}{2\Psi_{N}}\right)z^{\frac{\varphi}{2}}}{\sqrt{2\pi\Psi_{N}}A_{0}^{\varphi}\left(K_{\max}+1\right)!}\int_{\frac{\sqrt{z}}{A_{0}}}^{\infty}y^{-\varphi}\exp\left(-\frac{y^{2}}{2\Psi_{N}}\right)\nonumber\\&\times \left(\frac{t_{y}\sqrt{\mathcal{G}_{N}}}{\Psi_{N}}\right)^{K_{\max}+1}\exp\left(\frac{t_{y}\sqrt{\mathcal{G}_{N}}}{\Psi_{N}}\right)dy,
  \end{align}
  where $t_{y}\leq y$. Using this property along with $y^{-\varphi}\le\left(\sqrt{z}/A_{0}\right)^{-\varphi}$ and performing some mathematical operations, one gets\vspace{-0.25cm} 
  \begin{equation}
      \epsilon_{K_{\max}}(z)\leq  
 \frac{\mathcal{G}_{N}^{\frac{K_{\max}+1}{2}}\Psi_{N}^{-\frac{K_{\max}+1}{2}}}{\sqrt{2\pi}\left(K_{\max}+1\right)!}\underset{\triangleq\mathcal{J_{K_{\max}}}(z)}{\underbrace{\int_{\mathcal{U}(z)}^{\infty}\frac{\left(y+\sqrt{\frac{\mathcal{G}_{N}}{\Psi_{N}}}\right)^{K_{\max}+1}}{\exp\left(\frac{y^{2}}{2}\right)}dy}}.
  \end{equation}
where $\mathcal{U}(z)=\left(\sqrt{z}/A_{0}-\sqrt{\mathcal{G}_{N}}\right)/\sqrt{\Psi_{N}}$. It can be seen that\vspace{-0.25cm}
\[
\mathcal{J_{K_{\max}}}(z)\leq\left\{ \begin{array}{c}
\mathcal{V}_{K_{\max}},\mathcal{U}(z)\geq0\\
2\mathcal{V}_{K_{\max}},\mathcal{U}(z)<0
\end{array}\right.,
\]\vspace{-0.25cm}
where 
\begin{equation}  \mathcal{V}_{K_{\max}}=\sum_{i=0}^{K_{\max}+1}\frac{\left(K_{\max}+1\right)!2^{-\frac{i+1}{2}}}{\Gamma\left(\frac{i}{2}+1\right)\pi^{-\frac{1}{2}}\left(K_{\max}+1-i\right)!}\left(\frac{\mathcal{G}_{N}}{\Psi_{N}}\right)^{\frac{K_{\max}+1-i}{2}}.
\end{equation}
 Lastly, using the multinomial Theorem alongside \cite[Eq. (3.381.10)]{gradshteyn2007}
and the identity $i!=2^{i}\pi^{-\frac{1}{2}}\Gamma\left(\frac{i+1}{2}\right)\Gamma\left(\frac{i}{2}+1\right),$ it can be checked that\vspace{-0.2cm} 
\begin{equation}
   \frac{\mathcal{G}_{N}^{\frac{K_{\max}+1}{2}}\Psi_{N}^{-\frac{K_{\max}+1}{2}}}{\sqrt{2\pi}\left(K_{\max}+1\right)!}\times \mathcal{V}_{K_{\max}}=\mathcal{\mathcal{W}_{K_{\max}}}, 
\end{equation}
being defined in \eqref{Wkmax}. This completes the proof.\vspace{-.3cm}
  \section*{Appendix C\protect \protect \protect \\
		Proof of Proposition 2}
  Similarly to $h_{f,\ell}$, it can be seen by utilizing the CLT that $h_{f,e}^{(c)}$ and $h_{f,e}^{(s)}$ defined in \eqref{hfec}-\eqref{hfes} can be approximated by two normal distributions of zero-mean and same variance  as $
\mathbb{E}\left[\cos\left(\phi_{n}\right)\right]=\mathbb{E}\left[\sin\left(\phi_{n}\right)\right]=0$ and 
$\mathbb{E}\left[\left(h_{f,e}^{(c)}\right)^{2}\right]=\mathbb{E}\left[\left(h_{f,e}^{(s)}\right)^{2}\right]=\frac{N}{2}\mathcal{H}_{2}^{2}$. Thus, $\left|h_{f,e}\right|^{2}$ is a exponentially distributed of CDF\vspace{-0.15cm} 
\begin{equation} 
F_{\left|h_{f,e}\right|^{2}}\left(z\right)\approx1-\exp\left(-\frac{z}{N\mathcal{H}_{2}^{2}}\right).
\end{equation}\vspace{-.15cm}
Similarly to \eqref{CDFhl222}, the CDF of $\left\vert h_{e}\right\vert^{2}$ can be then evaluated as
\begin{align}
F_{\left\vert h_{e}\right\vert ^{2}}\left(z\right) & =\int_{0}^{A_{0}^{2}}F_{\left\vert h_{f,e}\right\vert ^{2}}\left(\frac{z}{y}\right)f_{ h_{p,e} ^{2}}\left(y\right)dy\nonumber \\
 & \overset{(a)}{\approx}1-\frac{\varphi}{2A_{0}^{\varphi}}\int_{\frac{1}{A_{0}^{2}}}^{\infty}\exp\left(-\frac{z}{N\mathcal{H}_{2}^{2}}t\right)t^{-\frac{\varphi}{2}-1}dt\nonumber \\
 & \overset{(b)}{\approx}1-\frac{\varphi}{2A_{0}^{\varphi}}\left(\frac{z}{N\mathcal{H}_{2}^{2}}\right)^{\frac{\varphi}{2}}\Gamma\left(-\frac{\varphi}{2},\frac{z}{N\mathcal{H}_{2}^{2}A_{0}^{2}}\right),
\end{align}
where steps (a) and (b) follow using the change of variable $y=1/t$
and \cite[Eq. (3.381.3)]{gradshteyn2007}, respectively. Subsequently, if $z\neq 0$, its PDF can be expressed as \vspace{-.05cm} 
\begin{align}
f_{\left\vert h_{e}\right\vert ^{2}}\left(z\right) & \approx \frac{\varphi}{2A_{0}^{\varphi}N\mathcal{H}_{2}^{2}}\int_{\frac{1}{A_{0}^{2}}}^{\infty}\exp\left(-\frac{z}{N\mathcal{H}_{2}^{2}}t\right)t^{-\frac{\varphi}{2}}dt,
\end{align}\vspace{-.05cm}
whereas the expression becomes for $z=0$ \vspace{-.05cm}
\begin{equation}
    f_{\left\vert h_{e}\right\vert ^{2}}\left(z\right)	\approx\frac{\varphi}{2A_{0}^{\varphi}N\mathcal{H}_{2}^{2}}\int_{\frac{1}{A_{0}^{2}}}^{\infty}t^{-\frac{\varphi}{2}}dt
	=\frac{\varphi}{A_{0}^{2}N\mathcal{H}_{2}^{2}\left(\varphi-2\right)}.
\end{equation}
Finally, \eqref{pdfhe222} is achieved by employing \cite[Eq. (3.381.3)]{gradshteyn2007}.\vspace{-.25cm}
\section*{Appendix D\protect \protect \protect \\
		Proof of Proposition 3}
  For a fixed $R_s$, the SOP is expressed using \eqref{eq:SNR} as\vspace{-.25cm}
  \begin{align}
\mathcal{P}_{\text{sop}}(R_s) & =\text{Pr}\left(\gamma_{\ell}\le2^{R_{s}}\left(\gamma_{e}+1\right)-1\right)\nonumber\\
 & =\text{Pr}\left(h_{\ell}^{2}\le\underset{\triangleq Z_{e}}{\underbrace{g\left(\left|h_{e}\right|^{2}\right)}}\right),\label{SOP0}
\end{align}\vspace{-.05cm}
where $g\left(z\right)=\mathcal{M}z+\mathcal{L}$ with
 $\mathcal{M}$ and $\mathcal{L}$ are being defined in Proposition 3.\vspace{-.35cm}
\subsection{Case 1: $R_s>0$}
In this case, we have\vspace{-.25cm}
\begin{equation}
     \mathcal{P}_{\text{sop}}(R_s) {=}\int_{0}^{\frac{1}{\mathcal{L}}}\frac{1}{z^{2}}F_{h_{\ell}^{2}}\left(\frac{1}{z}\right)f_{Z_{e}}\left(\frac{1}{z}\right)dz,\label{SOP00}
\end{equation}\vspace{-.45cm}
with\vspace{-.25cm}
\begin{align}
f_{Z_{e}}\left(\frac{1}{z}\right) & =\frac{1}{\mathcal{M}}f_{\left|h_{e}\right|^{2}}\left(\frac{\frac{1}{z}-\mathcal{L}}{\mathcal{M}}\right),z<\frac{1}{\mathcal{L}}.\label{Ze}
\end{align}
Substituting \eqref{eq:CDFhl} and \eqref{Ze} along with \eqref{pdfhe222} into \eqref{SOP00}, one obtains\vspace{-.25cm}
\begin{equation}
    \mathcal{P}_{\text{sop}}(R_s)=1-\left[\begin{array}{c}
\frac{1}{\mathcal{M}}\int_{0}^{\frac{1}{\mathcal{L}}}\mathcal{D}_{2}\left(\frac{1}{z}\right)dz-\frac{\left(2\Psi_{N}A_{0}^{2}\right)^{-\frac{\varphi}{2}}}{2\sqrt{\pi}\mathcal{M}}\\
\times\exp\left(-\frac{\mathcal{G}_{N}}{2\Psi_{N}}\right)\sum_{k=0}^{\infty}\int_{0}^{\frac{1}{\mathcal{L}}}\mathcal{D}_{1,k}\left(\frac{1}{z}\right)dz
\end{array}\right],\label{SOP1}
\end{equation}
where the functions $\mathcal{D_{\cdot}}$ are defined in \eqref{D1k}-\eqref{D2}. Finally, using the alternative extended Simpson's rule to the two inner integrals in \eqref{SOP1}, \eqref{SOP} for $\mathcal{L}>0$ is obtained.\vspace{-.45cm}
 \subsection{Case 2: $R_s=0$}
 \textcolor{black}{In} this case, the IP can be evaluated as\vspace{-.25cm} 
 \begin{align}
     \mathcal{P}_{\text{sop}}(0)&=\frac{1}{\mathcal{M}}\int_{0}^{\infty}F_{h_{\ell}^{2}}\left(z\right)f_{\left|h_{e}\right|^{2}}\left(\frac{z}{\mathcal{M}}\right)dz\nonumber\\
	&=1-\frac{1}{\mathcal{M}}\left[\begin{array}{c}
\int_{0}^{\infty}Q\left(\frac{\frac{\sqrt{z}}{A_{0}}-\sqrt{\mathcal{G}_{N}}}{\sqrt{\Psi_{N}}}\right)f_{\left|h_{e}\right|^{2}}\left(\frac{z}{\mathcal{M}}\right)dz\\
-\mathcal{B}_{N}\sum_{k=0}^{\infty}\int_{0}^{\infty}\mathcal{T}_{k}(z)f_{\left|h_{e}\right|^{2}}\left(\frac{z}{\mathcal{M}}\right)dz
\end{array}\right],
\end{align}
which leads to \eqref{SOP}, for $\mathcal{L}=0$, using Gauss-Laguerre \textcolor{black}{(GL)} quadrature.\vspace{-.25cm}
\color{black}
\section*{Appendix E\protect \protect \protect \\
		Proof of Proposition 4}
        When $P/\sigma_{x}^{2}\rightarrow\infty$, $\mathcal{L}$ approaches
$0$. Further, for equal value of $\sigma_{x}^{2}$ at nodes $\ell$
and $e$, $\mathcal{L}$ is the only parameter impacted by the SNR
value. Using Eqs. \eqref{SOP00}-\eqref{Ze}, the SOP can be asymptotically expressed,
using the integral's Taylor expansion at the first order, as 
    \begin{align}\color{black}
    \mathcal{P}_{\text{sop}}(R_{s})	&=\frac{1}{\mathcal{M}}\int_{\mathcal{L}}^{\infty}F_{h_{\ell}^{2}}\left(z\right)f_{\left|h_{e}\right|^{2}}\left(\frac{z-\mathcal{L}}{\mathcal{M}}\right)dz
	\nonumber \\&\sim\frac{1}{\mathcal{M}}\int_{0}^{\infty}F_{h_{\ell}^{2}}\left(z\right)f_{\left|h_{e}\right|^{2}}\left(\frac{z}{\mathcal{M}}\right)dz
	\nonumber \\&-\frac{\mathcal{L}}{\mathcal{M}}\left(\begin{array}{c}
F_{h_{\ell}^{2}}\left(0\right)f_{\left|h_{e}\right|^{2}}\left(0\right)\\
+\frac{1}{\mathcal{M}}\int_{0}^{\infty}F_{h_{\ell}^{2}}\left(z\right)f'_{\left|h_{e}\right|^{2}}\left(\frac{z}{\mathcal{M}}\right)dz
\end{array}\right).\label{eq:PSOPLetter}
\end{align}
Now, using Eq. \eqref{eq:CDFhl} along with the property $Q(x)=1-Q(-x)$, we get
\begin{equation}
F_{h_{\ell}^{2}}\left(0\right)\approx Q\left(\sqrt{\frac{\mathcal{G}_{N}}{\Psi_{N}}}\right).
\end{equation}
Subsequently, \eqref{eq:PSOPLetter} can be rewritten as
\begin{align}
\mathcal{P}_{\text{sop}}(R_{s}) & \sim\int_{0}^{\infty}\mathcal{K}\left(z\right)dz-\mathcal{L}\mathcal{O}_{N},\label{eq:SOP3}
\end{align}
where $\mathcal{O}_{N}$ is being defined in \eqref{ON} and $\mathcal{K}\left(\cdot\right)$
is expressed as 
\begin{equation}
\mathcal{K}\left(z\right)=F_{h_{\ell}^{2}}\left(\mathcal{M}z\right)\left(f_{\left|h_{e}\right|^{2}}\left(z\right)-\mathcal{L}f'_{\left|h_{e}\right|^{2}}\left(z\right)\right),\label{eq:Kz2}
\end{equation}
Using \eqref{pdfhe222}, the following holds 
\begin{align}
f'_{\left|h_{e}\right|^{2}}\left(z\right) & =\left[\left(\frac{\varphi}{2}-1\right)z^{-1}f_{\left|h_{e}\right|^{2}}\left(z\right)-\frac{\varphi z^{-1}}{2N\mathcal{H}_{2}^{2}A_{0}^{2}}\exp\left(-\frac{z}{N\mathcal{H}_{2}^{2}A_{0}^{2}}\right)\right]\label{eq:PDFprime}
\end{align}
Finally, substituting \eqref{eq:PDFprime} in \eqref{eq:Kz2}, \eqref{Kz}
is achieved. Further, by approximating the integral in \eqref{eq:SOP3}
with GL method, \eqref{SOPasympt0} is obtained. This concludes the proof of Proposition 4.
        \color{black}
\bibliographystyle{IEEEtran}
\bibliography{PaperFaissal}
\end{document}